\documentclass[conference]{IEEEtran}
\IEEEoverridecommandlockouts
\usepackage{cite}
\usepackage{amsmath,amssymb,amsfonts}
\usepackage[ruled, linesnumbered]{algorithm2e}
\usepackage{graphicx}
\usepackage{textcomp}
\usepackage{xcolor}
\usepackage{bm}
\usepackage[most]{tcolorbox}
\usepackage{threeparttable}
\usepackage{caption}
\usepackage{url}
\usepackage{amsmath}
\usepackage{booktabs}
\usepackage{tabularx}
\usepackage{xurl}
\usepackage{siunitx}
\newtheorem{theorem}{Theorem}
\newtheorem{proof}{Proof}
\usepackage{adjustbox}
\usepackage{multirow}
\usepackage{subcaption} 
\usepackage{wasysym}

\def\BibTeX{{\rm B\kern-.05em{\sc i\kern-.025em b}\kern-.08em
    T\kern-.1667em\lower.7ex\hbox{E}\kern-.125emX}}
\begin{document}

\title{PolyLink: A Blockchain Based Decentralized Edge AI Platform for LLM Inference\\
}
\IEEEoverridecommandlockouts

\author{\IEEEauthorblockN{Hongbo Liu\textsuperscript{1}, Jiannong Cao\textsuperscript{1}, Bo Yang\textsuperscript{2}, Dongbin Bai\textsuperscript{1}, Yinfeng Cao*\thanks{*Corresponding Author: Yinfeng Cao and Mingjin Zhang}\textsuperscript{1}, Xiaoming Shen\textsuperscript{1}, Yinan Zhang\textsuperscript{1}\\Jinwen Liang\textsuperscript{1}, Shan Jiang\textsuperscript{3}, Mingjin Zhang*\textsuperscript{1}}
\IEEEauthorblockA{
\textsuperscript{1}\textit{Department of Computing, The Hong Kong Polytechnic University, Hong Kong, China}\\
\textsuperscript{2}\textit{China Mobile (Hong Kong) Innovation Research Institute, Hong Kong, China}\\
\textsuperscript{3}\textit{School of Software Engineering, Sun Yat-sen University, China}\\
hong-bo.liu@connect.polyu.hk, csyfcao@comp.polyu.edu.hk
}
}



\maketitle

\begin{abstract}
The rapid advancement of large language models (LLMs) in recent years has revolutionized the AI landscape. However, the deployment model and usage of LLM services remain highly centralized, creating significant trust issues and costs for end users and developers. To address these issues, we propose PolyLink, a blockchain-based decentralized AI platform that decentralizes LLM development and inference. Specifically, PolyLink introduces a decentralized crowdsourcing architecture that supports single-device and cross-device model deployment and inference across heterogeneous devices at the edge. Moreover, to ensure the inference integrity, we design the TIQE protocol, which combines a lightweight cross-encoder model and an LLM-as-a-Judge for a high-accuracy inference evaluation. Lastly, we integrate a comprehensive token-based incentive model with dynamic pricing and reward mechanisms for all participants.
We have deployed PolyLink and conducted an extensive real-world evaluation through geo-distributed deployment across heterogeneous devices. Results indicate that the inference and verification latency is practical. 
Our security analysis demonstrates that the system is resistant to model degradation attacks and validator corruptions.
PolyLink is now available at \url{https://github.com/IMCL-PolyLink/PolyLink}.

\end{abstract}


\section{Introduction}

\noindent Artificial Intelligence (AI) has experienced massive growth and adoption across various domains.
In particular, cloud-based Large Language Models (LLMs) such as ChatGPT, Claude, and Gemini have emerged as groundbreaking AI services in recent years, demonstrating remarkable capabilities in understanding, generation, and reasoning in general tasks, thus enabling a wide range of AI applications.


However, current AI services and applications are highly centralized, with few stakeholders (e.g. cloud service providers) controlling the majority of infrastructure, models, and access. This centralization arises from several practical factors: both training and inference of large models require substantial computational and energy resources, which are typically concentrated in data centers. Normal end users, small/middle enterprises and organizations often cannot afford such hardware and software costs.
Consequently, the centralized nature of current AI creates significant barriers to access and improve AI, particularly for individuals and organizations in developing regions or with limited financial resources.

AI democratization is a recent new trend aiming at addressing the AI centralization issues by redistributing AI from centralized stakeholders to decentralized parties \cite{10.1145/3696410.3714666}. Specifically, the democratization of AI contains three perspectives: 1) usage: making the price of AI services and applications more accessible and affordable to end users; 2) development: promoting cost-effective hardware infrastructure such as GPUs and open-source AI development framework for developers to freely develop and deploy their customized AI models (e.g., federated learning \cite{cao2021toward}); 3) governance: enabling distributed ownership and transparency management on the AI models and hardware infrastructure. 


Decentralized Physical Infrastructure Networks (DePIN) are regarded as a promising solution to realize AI democratization \cite{lin2024decentralized}. In DePIN protocols, participants contribute their idle computational resources (e.g. GPUs and CPUs) to a shared network, which are further clustered for support AI model training and inference. Blockchain-based incentive such as tokens are used to reward the contributors to govern the network that incentivize the high-quality AI services with integrity. However, existing DePIN protocols suffer from three challenges. 
First, the limited computational resource provided by low-end devices makes it difficult for DePINs to support large-scale LLMs.
Second, the verification mechanisms for the integrity of computational results are either insecure or inefficient, which is vulnerable to malicious device owners that try to perform dishonest computations for profit. 
Third, the existing DePIN incentive mechanisms is typically less effective as they only apply the similar traditional cloud service price model (device provider and users), overlooking the rewards for developers who contribute AI models.

To address these challenges, we propose \textsc{PolyLink}, a blockchain-based decentralized AI platform running over edge networks. We develop an inference framework that supports both single-device and cross-device execution, enabling flexible deployment across heterogeneous devices with varying model sizes and computational demands.
Then we design a Trustless Inference Quality Evaluation (TIQE) protocol that ensures the inference integrity without centralized authority. 
At last, we develop a comprehensive incentive mechanism that fairly rewards device contributors and model providers based on their contributions and result quality, while charging fees to service users.
We conduct extensive real-world deployment and evaluation across 20 geo-distributed devices contributed by multiple universities. Results show that the system delivers robust inference with acceptable latency across heterogeneous edge environments. The TIQE protocol achieves a favorable balance between low latency and high accuracy. Our security analysis further demonstrates resilience against model degradation and validator corruption attacks.
The main contributions of this paper are as follows:

\begin{itemize}
\item We design \textsc{PolyLink}, the first blockchain-based decentralized AI platform with full design details that supports both single-device and cross-device LLM inference over heterogeneous edge networks.
\item We propose the TIQE protocol, combining lightweight Cross-Encoder model with LLM-as-a-Judge to ensure inference integrity with low cost.
\item We develop a dynamic incentive model with rewards and pricing mechanisms to align the interests of model providers, workers, and validators.
\item We evaluate \textsc{PolyLink} through real-world large-scale geo-distributed deployment on diverse devices and provide security analysis showing resilience to model degradation and validator corruption.
\end{itemize}

\begin{figure*}[htb]
    \centering
    \includegraphics[width=0.75\linewidth]{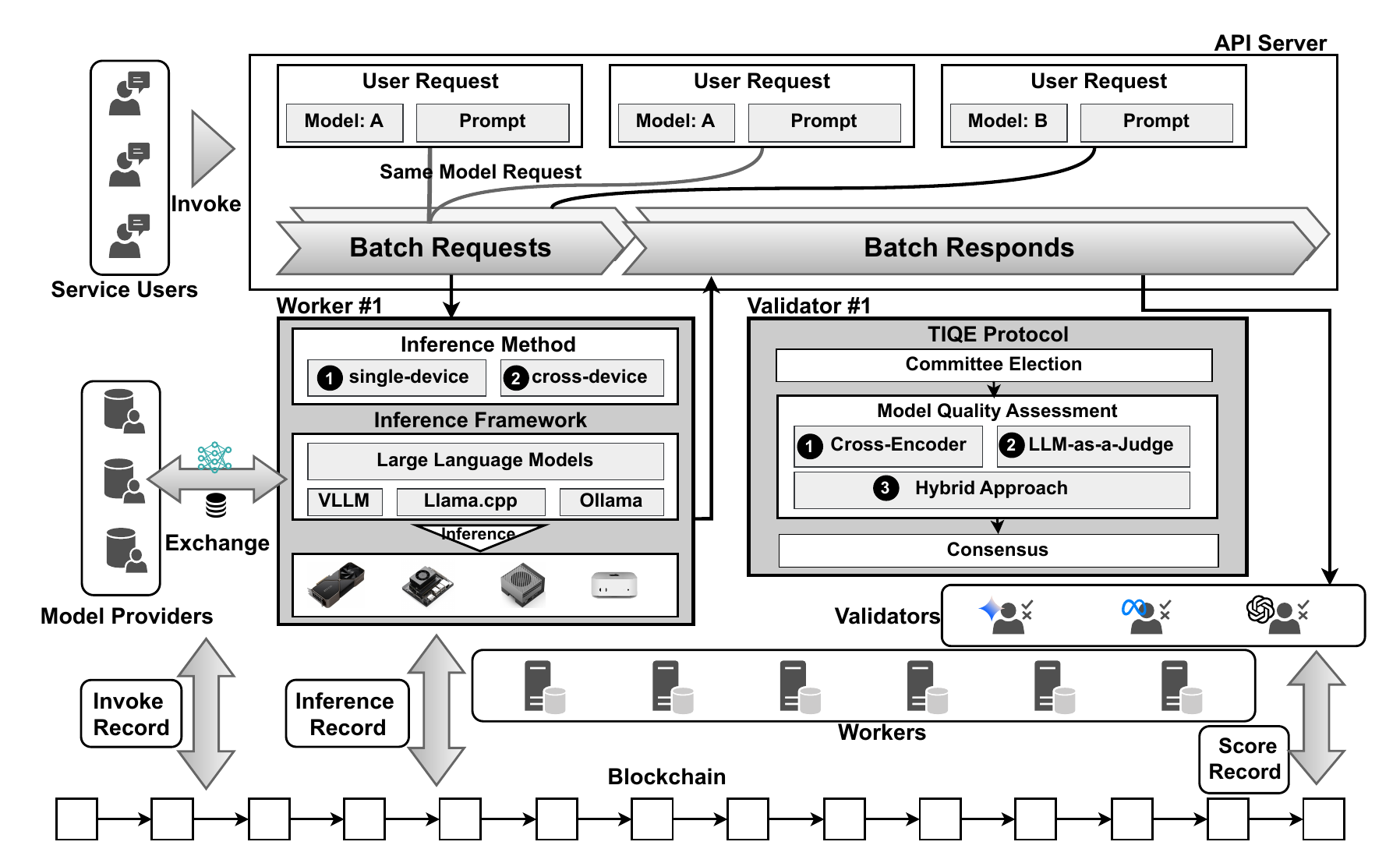}
    \caption{Overview of \textsc{PolyLink}, a blockchain-based decentralized edge AI platform. Service users send inference requests via the API. The API server batches requests targeting the same model, which are then processed by a worker running that model. The responses are returned and evaluated for inference quality by the validators.}
    \label{fig:overview}
\end{figure*}

\section{Related Work}
\noindent We conduct a comprehensive review of existing trustless inference protocols and decentralized AI platforms, summarizing their key contributions and analyzing their limitations (Tab. \ref{tab:related}).

\noindent \textbf{Trustless Inference Protocol.} Trustless inference protocols aim to guarantee the integrity of LLM inference performed on untrusted devices. zkLLM \cite{sun2024zkllm} successfully provides cryptographically verifiable inference by generating a zero-knowledge proof of each result with fully decentralization. However, the method incurs substantial overhead. For LLaMa-2-13B, it requires \SI{986}{\second} to generate the model commitment and \SI{803}{\second} to generate a proof for every inference. SVIP \cite{sun2024svip} adopts an activation-based approach that trains a classifier on final-layer activations and inputs, yet it depends on a trusted third party (TTP) for training and lacks generalizability. TOPLOC \cite{ong2025toploc} and SPEX \cite{dallachiesa2025statistical} use locality-sensitive hashing (LSH) to hash activations and verify inference integrity, but they require re-inference by TTP and cannot efficiently distinguish high-precision model executions.

\noindent \textbf{Decentralized AI platform.} Recently, numerous projects have pursued decentralized AI within the Web3 ecosystem. AIArena \cite{wang2025aiarena} targets decentralized training: validators use public datasets to evaluate training quality. io.net \cite{io_net} provides decentralized computing resources and relies on device-side status monitors and work logs to ensure that workers execute tasks honestly. SMART \cite{huang2024advancing} offers a hybrid on-chain/off-chain inference framework, but its dependence on trusted execution environments (TEEs) constrains the platform’s scalability and efficiency.

\begin{table}[t]
  \centering
  \caption{Summary of Existing Protocols and Platforms}
  \label{tab:related}
  \newcommand*\feature[1]{\ifcase#1 \Circle\or\LEFTcircle\or\CIRCLE\fi}
  \newcommand*\f[3]{\feature#1 & \feature#2 & \feature#3}
  
  \begin{threeparttable}
  \resizebox{\linewidth}{!}{%
  \begin{tabular}{lccccc}
    \toprule
    Category & Work  & Service & Decentralization & Integrity & Efficiency\\
    \midrule
    \multirow{4}{*}{Protocol} 
      & zkLLM~\cite{sun2024zkllm}   & Inference  & \feature{2} & Cryptography-based & \feature{0} \\
      & SVIP~\cite{sun2024svip}     & Inference & \feature{1} & Learning‐based & \feature{1} \\
      & TOPLOC~\cite{ong2025toploc} &  Inference & \feature{1} & LSH‐based & \feature{1} \\
      & SPEX~\cite{dallachiesa2025statistical}  & Computation & \feature{1} & LSH‐based & \feature{1}  \\
    \midrule
    \multirow{4}{*}{Platform}
      & AIArena~\cite{wang2025aiarena}  & Training & \feature{2} & Evaluation-based & \feature{1} \\
      & io.net~\cite{io_net}    & Computation  & \feature{1} & Log-based & \feature{2} \\
      & SMART~\cite{huang2024advancing} & Inference & \feature{2} & TEE-based & \feature{1} \\
      & \textsc{PolyLink}   & Inference & \feature{2} & Evaluation-based & \feature{1} \\
    \bottomrule
  \end{tabular}
  }
    \begin{tablenotes}\footnotesize
    \item $\dagger$ $\feature2$ = provides property; $\feature1$ = partially provides; $\feature0$ = does not provide. 
  \end{tablenotes}
  \end{threeparttable}
\end{table}

\section{PolyLink Overview}

\subsection{System Model}
\noindent Fig. \ref{fig:overview} shows the overview of \textsc{PolyLink} system. There are several kinds of participants in the \textsc{PolyLink}:

\begin{itemize}
    \item \textbf{Service User:} Service users send inference requests to the system and pay with cryptocurrency tokens for computation.
    \item \textbf{Worker:} Participants with idle GPU resources, involving NVIDIA GPU, NVIDIA Jetson, Apple Silicon etc., contribute those resources and earn rewards. Workers are divided into two categories: \emph{worker with single device} and \emph{worker with multiple devices}. 
    \item \textbf{Model Provider:} Model providers are responsible for deploying their own LLMs, such as Fine-tuning Models, on the workers and earn rewards.
    \item \textbf{Validator:} Validators are responsible for evaluating the inference quality responded by model on worker. To participate, validators must stake significant tokens.
\end{itemize}

\subsection{Threat Model}
\noindent \textsc{PolyLink} is a decentralized platform, most system participants are untrusted, which may introduce the following threats.

\noindent \textbf{Model Degradation Attack.} A malicious worker may use low-precision models or perform lazy computation for profit, leading to degraded inference quality.

\noindent \textbf{Validator Corruption Attack.} A malicious validator may intentionally submit manipulated scores (either excessively high or low) to disrupt the consensus process, distort quality evaluation, and undermine the fairness of reward distribution.
    


We assume that models from providers are benign, the underlying blockchain operates correctly and remains available, and fewer than $1/3$ of validators are malicious, following the standard Byzantine Fault Tolerance assumption.

\subsection{Design Objectives}
\noindent To tackle the challenges in DePINs, the design objectives of \textsc{PolyLink} are as follows.

    \noindent \textbf{Decentralization and Trustless.}  The system coordinates computation and validation among participants without relying on TTP. Both the validator committee and model deployment follow a decentralized, edge-computing paradigm, distinguishing the system from traditional cloud platforms such as OpenAI or Google. Moreover, we prioritize the devices that are actually geo-distributed.

    \noindent \textbf{Inference Integrity.}  Inference requests are executed honestly on decentralized workers, and their results are evaluated by validators through a protocol to guarantee integrity.

    \noindent \textbf{Incentive and Protocol Efficiency.}  
    The system is designed to ensure fair and sustainable incentives for each entity. Moreover, the protocol employed by the system for integrity does not introduce substantial overhead.

\section{PolyLink Specification}

\subsection{Decentralized Model Inference}
\noindent Our platform enables decentralized model inference by routing all inference requests to workers. We provide two categories of inference on workers: \emph{single-device} and \emph{cross-device} inference.

\noindent \textbf{Single-device Inference.} Model $M$ is deployed on a single device. Typically, this approach is used for models with relatively few parameters. We define the inference process as a function 
\begin{equation}
M(t^p) \rightarrow t^r
\end{equation}
where a service user inputs a prompt $t^p$ of length $p$ into a model $M$. The model outputs an answer $t^r$ of length $r$.

\noindent \textbf{Cross-device Inference.} For workers with multiple devices aim to deploy models with large parameters, we adopt the \emph{EdgeShard} scheme \cite{zhang2024edgeshard}, which partitions the model $M$ into $n$ sequential shards executed on an independent device:
\begin{equation}
M = \prod_{i=1}^nM_{i}
\end{equation}

Due to the autoregressive nature of LLMs, inference is performed iteratively for $r$ rounds across shards. Let $x_{0;0} := t^p$ denote the initial input. In the round $k(k \geq 1)$, for each shard $M_i$, the intermediate output is computed as:
\begin{equation}
x_{i;k} \leftarrow \left\{
\begin{array}{rcl}
M_i(x_{n;k-1}) & & {i = 1}\\
M_i(x_{i-1;k}) & & {i > 1}\\
\end{array} \right.
\end{equation}

After $r$ rounds, the final response is:
\begin{equation}
t^r = \sum_{k=1}^rx_{n;k}
\end{equation}

\subsection{Pricing and Incentive Mechanism}
\subsubsection{Pricing Mechanism}

For each inference request $Req_i$, the cost $C_{\text{inference}}^i$ is determined by the input/output token lengths $p,r$ and the computational complexity of the model. The pricing formula is defined as: $ C_{\text{inference}}^i = \delta \cdot S_M \times (C_{\text{in}}^i \times p + C_{\text{out}}^i \times r) $ where $S_M$ denotes the normalized model scale factor, $\delta$ is a scaling coefficient that reflects market-based adjustments. $C_{\text{in}}^i$ and $C_{\text{out}}^i$ are the base cost per input token and output token.

\subsubsection{Incentive Mechanism}
In each inference batch for a selected model run on one worker, we assume a base reward per inference task is denoted as $R_i=\theta \cdot C_{inference}^i$, $\theta$ is the reward parameter. A batch contains $b$ inference tasks, the total reward allocated to that batch is $ R_{\text{batch}} = \sum_{i=1}^b R_i $. 
The reward is divided into two parts: a fixed portion $\beta \cdot R_{\text{batch}}$ allocated to validators, and a dynamic portion $(1 - \beta) \cdot R_{\text{batch}}$ allocated based on the model quality score $\alpha \in [0,1]$ (Algorithm. \ref{alg:rdib}).


\noindent \textbf{Worker Reward.} The worker receives a share of the dynamic reward based on its model score:

\begin{equation}
R_{\text{worker}} = \alpha (1 - \beta) \cdot R_{\text{batch}}
\end{equation}

\noindent \textbf{Validator Reward.} Let there be $m$ validators in the batch. Each validator $v_i$ stakes an amount of token $ST_{v_i}$. Define the vector of staked tokens as:

\[
\bm{ST} = \{ ST_{v_1}, ST_{v_2}, \ldots, ST_{v_m} \}
\]

The total validator reward includes both the fixed base and the remainder of the dynamic portion:

\begin{equation}
R_{\text{validators}} = \left[ \beta + (1 - \alpha)(1 - \beta) \right] \cdot R_{\text{batch}}
\end{equation}

Each validator $v_i$ receives a share of the validator reward proportional to their staked amount:

\begin{equation}
R_{v_i} = \frac{ST_{v_i}}{\sum_{j=1}^{m} ST_{v_j}} \cdot R_{\text{validators}}
\end{equation}


\begin{algorithm}[htbp]
\caption{Reward Distribution for Inference Batch}\label{alg:rdib}
\KwIn{
  $\bm{R}$: reward per inference task vector\\
  $b$: number of inference tasks in batch \\
  $\alpha$: model quality score \\
  $\beta$: validator base reward factor \\
  $\bm{ST}$: validator stakes \\
    }
\KwOut{
  $R_{\text{worker}}$: reward for worker \\
  $\{R_{v_1}, \ldots, R_{v_m}\}$: rewards for validators
}
\SetKwProg{Fn}{Procedure}{:}{End}
\Fn{RewardDistributed $(\bm{R}, b,\alpha,\beta,\bm{ST})$}{
$R_{\text{batch}} := 0$ 

\ForEach{$R_i \in \bm{R}$}{
  $R_{\text{batch}} := R_{\text{batch}} + R_i$
}

$R_{\text{worker}} := \alpha \cdot (1 - \beta) \cdot R_{\text{batch}}$

$R_{\text{validators}} := [\beta + (1 - \alpha) \cdot (1 - \beta)] \cdot R_{\text{batch}}$ 

\For{$j \leftarrow 1$ \KwTo $m$}{
  $ST_{total} := ST_{total}+ST_{v_j}$ 
}
\For{$i \leftarrow 1$ \KwTo $m$}{
  $R_{v_i} := \frac{ST_{v_i}}{ST_{total}} \cdot R_{\text{validators}}$ 
}
}
\Return $R_{\text{worker}}, \{R_{v_1}, \ldots, R_{v_m}\}$ 
\end{algorithm}

\noindent \textbf{Model Provider Reward.} The model provider receives rewards through transactions with the worker. Specifically, when a worker selects and utilizes a model provided by a model provider, it must pay a \emph{model usage fee (MUF)} defined by the model provider. This fee can be fixed or dynamically priced based on model quality. The worker can also act as the model provider. 

\subsection{Trustless Inference Quality Evaluation Protocol}

\noindent To ensure inference integrity and quality, we propose a \emph{Trustless Inference Quality Evaluation (TIQE)} protocol (Fig. \ref{fig:epoch}), which enables a decentralized validator committee to perform quality evaluations and reach consensus on the results. Based on the consensus, a \emph{model quality score} is assigned to the model run on the worker.

\begin{figure}[htbp]
    \centering
    \includegraphics[width=0.95\linewidth]{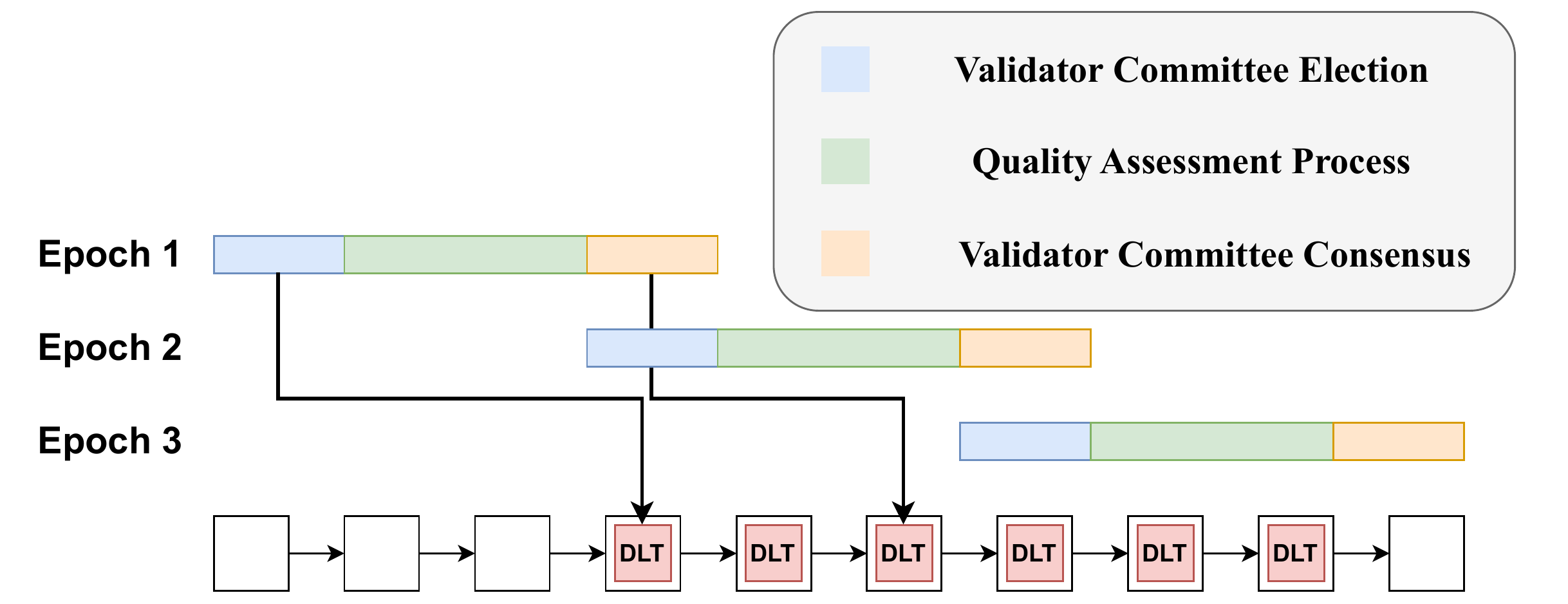}
    \caption{Overview of the \emph{TIQE} Protocol in each Epoch}
    \label{fig:epoch}
\end{figure}

\subsubsection{Quality Assessment Process}
In this process, validators evaluate inference results using a trustless scoring mechanism. The evaluation function is formally defined as:

\begin{equation}
\text{eval}(t^p, t^r, \mathcal{J})  \rightarrow \mathcal{J}(t^p,t^r) \rightarrow \textit{Score}
\end{equation}

where $\mathcal{J}$ indicates the judgment approach applied to assess the result. We propose three types of judgment approaches to support efficient, cost-effective evaluation of inference quality.

\noindent \textbf{Cross-encoder Approach.} 
Cross-encoder is a light-weight Transformer-based model (e.g. BERT, LLM with few parameters) trained to measure the similarity between related query-document pairs \cite{reimers2019sentence,zhang2024proof}. We construct an input pair by concatenating the prompt and the corresponding output result, which is then fed into a cross-encoder model to produce a quality score. This cross-encoder provides a lightweight evaluation solution that estimates the semantic similarity between the prompt and the result. However, it is limited in open generative scenario since it's limitation model capability. 

\noindent \textbf{LLM-as-a-Judge Approach.}
With the significant generalization and reasoning capabilities, LLM is widely used to evaluate the performance of other LLMs \cite{gu2024survey}. Typically, the judge LLM must be a large-scale model with strong reasoning capabilities, such as ChatGPT-o3 (175B), Gemini 2.5 Pro, DeepSeek-V2 (671B), or LLaMA 3.2 (405B). Figure~\ref{fig:prompt} illustrates an example of the evaluation prompt used by the judge LLM to assess the output quality and produce a corresponding model score.

\begin{figure}[htbp]
    \centering
    \begin{adjustbox}{width=0.85\linewidth}  
    \begin{tcolorbox}
You are an expert judge. Your task is to rate the quality of the following LLM inference result given the provided input.  
Rate on a scale from 1 to 5, where:

1 = Completely incorrect or nonsensical\\
2 = Mostly incorrect or with major flaws\\
3 = Partially correct but with noticeable issues\\
4 = Mostly correct with minor issues\\
5 = Completely correct, comprehensive, and well-reasoned

\textit{Input:}\\
\{input\}

\textit{LLM Inference Output:}\\
\{output\}

Please return only the numeric score (1 to 5) and no explanation.

\textit{Score:}
    \end{tcolorbox}
    \end{adjustbox}
    \caption{The prompt of \emph{LLM-as-a-Judge} to evaluate the decentralized LLM inference quality.}
    \label{fig:prompt}
\end{figure}

The \textit{LLM-as-a-Judge} approach can be implemented either by invoking an commercial API or by deploying the judge model on validator nodes. However, invoking commercial APIs, such as OpenAI's GPT-o3 (which charges up to \$40.0 per million tokens for output \footnote{\url{http://openai.com/api/pricing/}}), introduces significant costs. Alternatively, deploying judge models locally imposes substantial hardware requirements on validators, including high memory capacity and GPU resources.

\noindent \textbf{Hybrid Approach.}  
To mitigate the limitations of both the \emph{cross-encoder} and \emph{LLM-as-a-Judge} methods, we propose a hybrid evaluation approach. Within each epoch, when validators collect a batch of inference tasks associated with a specific model on a worker, they evaluate this batch using a cross-encoder. In epoch $e$, the cross-encoder score $Score_{cross}^e$ is defined as:

\begin{equation}
\text{Score}_{\text{cross}}^e = \frac{1}{\lvert \mathcal{B} \rvert} \times \sum_{B \in \mathcal{B}} \sum_{i \in B} eval(t^p_i,t^r_i,\{ \textit{Cross-Encoder}\})
\end{equation}

At a randomly selected point within the epoch, an evaluation score using the \emph{LLM-as-a-Judge} approach, which is assigned a higher weight in the final scoring computation. In epoch $e$, the LLM score $Score_{llm}$ is defined as:

\begin{equation}
\text{Score}_{\text{llm}}^e =  eval(t^p_e,t^r_e,\{ \textit{LLM}\})
\end{equation}

This hybrid approach allows the system to maintain evaluation accuracy while significantly reducing computational costs, thereby providing a more cost-effective and scalable solution for trustless model quality evaluation. After a specific epoch $e$, the final model quality score is defined as: 

\begin{equation}
\text{Score}^e_{\text{final}} = \lambda \cdot \text{Score}^e_{\text{llm}} + (1 - \lambda) \cdot \text{Score}^e_{\text{cross}}
\label{eq:11}
\end{equation}

\subsubsection{Model Quality Score Consensus}

In the \emph{TIQE} protocol, validators are elected at the beginning of each epoch, and they are responsible for reaching consensus on the model quality score. In this section, we introduce the design of these two key processes: validator election and consensus formation.

\noindent \textbf{Validator Committee Election.} At the beginning of each epoch, a validator committee is elected using a Verifiable Random Function (VRF)-based selection mechanism. When a validator $V_k \in \mathcal{V}$ observes the start of a new epoch $e$, it performs a VRF using its secret key $sk_k$ and a random seed which is composed by previous block hash $H_{\text{prev}}$ and epoch hash $ \mathcal{H}(e)$.

\begin{equation}
(r_k, \pi_k) = \text{VRF}_{sk_k}(H_{\text{prev}} \parallel \mathcal{H}(e))
\end{equation}

where $r_k$ is a verifiable pseudorandom and $\pi_k$ is a proof that $r_k$ was correctly computed with $sk_k$. Our system defines a public election rule $\mathcal{R}(r_k)$:

\begin{equation}
\mathcal{R}(r_k): \quad r_k \bmod \lvert \mathcal{V} \rvert < \lvert M_{\text{committee}} \rvert
\end{equation}

where $\lvert \mathcal{V} \rvert$ is the total number of eligible validators, $M_{\text{committee}}$ is the desired committee and $\lvert M_{\text{committee}} \rvert$ is the committee size. If the $r_k$ satisfies the rule, the validator becomes a member of the validator committee $M_{committee}^{(e)}$ for epoch $e$. The proof $\pi_k$ can be used by any other node to verify that $V_k$ was elected fairly, without revealing its secret key $sk_k$. This mechanism ensures decentralized, unpredictable, and verifiable election of validators in each epoch.

\noindent \textbf{Validator Committee Consensus.} In the consensus phase of each epoch, all validators in the elected committee submit their individual model quality scores to the smart contract $SC$. After $SC$ collectes sufficient scores $\bm{Score}$, $SC$ computes the median score $Score_{\text{med}}$ from $\bm{Score}$.

To discourage dishonest or biased evaluations, we define a deviation threshold $Th$. If a validator's score $Score_i \in \bm{Score}$ deviates from the median score by more than $Th$, i.e.,

\begin{equation}
|Score_i - Score_{\text{med}}| > Th
\label{eq:14}
\end{equation}

the validator is penalized by slashing a portion of staked tokens. The smart contract shown as Algorithm \ref{alg:vcc}. 

\begin{algorithm}[htbp]
\caption{Scoring Consensus Smart Contract}\label{alg:vcc}
\KwIn{
  $\bm{Score}$ : scores submitted by validators\\
  $\bm{ST}$: staked tokens of validators\\
  $Th$: deviation threshold
}
\KwOut{
  $Score_{\text{med}}$: final consensus model score\\
  $\bm{Slashed}$: slashing vector for penalized validators
}
\SetKwProg{Fn}{Procedure}{:}{End}
\Fn{ConsensusOnScore $(\bm{Score}, \bm{ST}, Th)$}{
  Sort $\bm{Score}$ in ascending order 
  $Score_{\text{med}} := \text{Median}(\bm{Score})$ 
  $\bm{Slashed} := \{0, \ldots, 0\}$ 
  \ForEach{$i \leftarrow 1$ \KwTo $m$}{
    \If{$|Score_i - Score_{\text{med}}| > Th$}{
       $\bm{Slashed}[i] \leftarrow \gamma \cdot ST_{v_i}$ \\
       \tcp{$\gamma$ is the slashing rate}
    }
  }
}
\Return $Score_{\text{med}}, \bm{Slashed}$ 
\end{algorithm}

\section{Security Analysis}
\noindent In this section, we analyze the security of \textsc{PolyLink} under the threat model. 



\begin{theorem}[Model Degradation Attack Resistance]
If a worker continuously performs low-precision inference, the reward obtained in the designated epoch $e$ will be 0.
\end{theorem}

\begin{proof}
The worker's reward in epoch $e$ is $R_\text{worker} = \alpha^e (1 - \beta) \cdot R_\text{batch}$, where $\alpha^e \in [0, 1]$. Assume the worker performs continuous low-precision inference in previous $e-1$ epochs. Then, according to the evaluation function (Eq.~\ref{eq:11}), the final quality score satisfies $\alpha_e = \lambda \cdot Score_\text{LLM}^e + (1 - \lambda) \cdot Score_\text{Cross}^e \approx 0$. Thus,$R_\text{worker} \sim 0 \cdot (1 - \beta) \cdot R_\text{batch} = 0$. Therefore, a worker conducting consistently low-quality inference will receive zero reward.
\end{proof}

\begin{theorem}[Validator Corruption Attacks Resistance]
If a validator continuously performs dishonestly ($< 1/3$), the staked token in the designated epoch $e$ will be slashed to 0.
\end{theorem}



\begin{proof}
Let $Score_i$ be the score submitted by validator $v_i$ in epoch $e$, and let $Score_\text{med}$ be the committee's median score. The smart contract enforces slashing if the deviation exceeds a threshold $Th$: $|Score_i - Score_\text{med}| > Th \Rightarrow \bm{Slash}[i] = \gamma \cdot ST_{v_i}, \gamma \in (0,1]$. We consider three typical dishonest behaviors:
\begin{itemize}
    \item \textbf{Over-scoring.} Validator gives artificially high scores to low-quality outputs. As the majority of validators are assumed honest ($<1/3$ are malicious), the honest scores anchor $Score_\text{med}$ closer to the ground truth. Hence, inflated scores from $v_i$ will exceed the threshold and trigger slashing.

    \item \textbf{Under-scoring.} Validator gives abnormally low scores to degrade honest workers' scores. Again, the deviation from $Score_\text{med}$ (which is close to honest average) leads to $|Score_i - Score_\text{med}| \gg Th \Rightarrow \text{Slashed}$
    
    \item \textbf{Random-scoring.} Validator submits inconsistent or uncorrelated scores to disrupt consensus. Statistically, such behavior will periodically deviate beyond $Th$ and accumulate slashing over epochs.
\end{itemize}
In each case, the validator’s staked token is slashed by a rate $\gamma$ per dishonest epoch. After $k$ dishonest epochs, the remaining stake is: $ ST_{v_i}^{(k)} = ST_{v_i}^{(0)} \cdot (1 - \gamma)^k$. Taking the limit as $k \to \infty$: $\lim_{k \to \infty} ST_{v_i}^{(k)} = 0$. Therefore, any validator that persistently submits dishonest scores will eventually lose all staked tokens.
\end{proof}

\section{Implementation and Evaluation}

\subsection{Implementation}\label{AA}

\noindent We implement the proposed \textsc{PolyLink} system, integrating all previously described components. The system backend and frontend are deployed on a cloud server instance with 4-core CPU and 8G memory. We deploy the smart contracts and issue the ERC-20 token on the Sepolia testnet\footnote{\url{https://sepolia.etherscan.io/address/0x9711b259e6281a1eA9465362Cb0BDd5D9Bf35AaD}}.

\subsection{Real-world Evaluation}
\subsubsection{Geo-distributed Deployment} 

We perform a large-scale geo-distributed deployment of 20 devices from 10 different workers across multiple regions, including Hong Kong SAR, Guangzhou and Shenzhen in China and Kanazawa in Japan. 

These workers operate in a decentralized edge network configuration, where each worker runs independently and participates in trustless inference tasks. The detailed hardware specifications of deployed worker nodes are shown in Table~\ref{tab:edge_nodes}.











\begin{table}[htbp]
\centering
\caption{Hardware Specifications of Geo-distributed Edge Workers}
\label{tab:edge_nodes}
\begin{adjustbox}{width=\linewidth}
\begin{tabular}{llp{8cm}c}
\toprule
\textbf{Worker} & \textbf{Location} & \textbf{Device Types} & \textbf{Quantity} \\
\midrule
1 & Guangzhou, Panyu & NVIDIA RTX 2080 Ti & 2 \\

2 & Kanazawa, Kakumamachi & NVIDIA RTX A4500 & 1 \\

\multirow{2}{*}{3} 
 & Hong Kong, Hung Hom & NVIDIA RTX 4060 Ti$\times$1, RTX 3080 Ti$\times$1, Jetson Orin NX 16GB$\times$3, Jetson AGX Orin 32GB$\times$3 & 8 \\
 & Shenzhen, Luohu & Apple MacBook Pro (M3 Pro) & 1 \\

4 & Hong Kong, Pokfulam & NVIDIA Tesla P100 & 2 \\

5 & Hong Kong, Sai Kung & NVIDIA RTX 4090 & 8 \\

6 & Hong Kong, Hung Hom & NVIDIA RTX 3090 & 3 \\

7 & Hong Kong, Sha Tin & NVIDIA RTX 3090 & 4 \\

8 & Hong Kong, Sha Tin & NVIDIA RTX 3060 & 1 \\

9 & Hong Kong, Hung Hom & NVIDIA RTX 2060$\times$1, RTX 3090$\times$5, RTX 3090 Ti$\times$1 & 7 \\

10 & Hong Kong, Hung Hom & NVIDIA Quadro GV100$\times$2, Tesla A100$\times$1, RTX 4090$\times$4 & 7 \\
\bottomrule
\end{tabular}
\end{adjustbox}
\end{table}

\subsubsection{Evaluation Metrics} 
We conduct evaluations from three primary aspects: \emph{Geo-distributed Decentralized Inference Performance}, \emph{TIQE Protocol Performance}, and \emph{Reward Distribution}.
        
\noindent \textbf{Geo-distributed Decentralized Inference Performance.} We consider several metrics: \textit{Average Latency} and \textit{Time to First Token (TTFT)} to measure user experience, \textit{Output Token Throughput} and \textit{Request Throughput} to assess system performance, and \textit{Failure Rate} to reflect system stability.

\noindent \textbf{TIQE Protocol Performance.} We evaluate the TIQE protocol using two indicators:

\begin{itemize}
    \item \textbf{Performance Overhead:} This refers to the computational and financial cost introduced by the quality assessment process. Specifically, we measure the latency of the Cross-Encoder and LLM-as-a-Judge components, as well as the cost associated with invoking LLM-as-a-Judge.
    
    \item \textbf{Degradation Detection Accuracy:} This evaluates the effectiveness of TIQE in identifying degraded outputs. We report the True Positive (TP) rate and False Positive (FP) rate for both the Cross-Encoder and LLM-as-a-Judge.
\end{itemize}

\noindent \textbf{Task Reward Distribution.} We consider reward distribution under different conditions to evaluate:
\begin{itemize}
    \item \textbf{Incentive Effectiveness:} whether high-quality performers receive appropriate rewards.
    \item \textbf{Penalty Rationality:} whether low-quality or dishonest behaviors are properly penalized.
    \item \textbf{Fairness of Distribution:} whether validators receive fair rewards proportional to their stake.
\end{itemize}

\subsection{Experimental Setup}

\subsubsection{Geo-distributed Decentralized Inference Performance} 
We deploy LLMs with different parameters on the workers: DeepSeek-R1-1.5B, DeepSeek-R1-7B, and DeepSeek-R1-14B. 
The models are executed using both single-device and cross-device settings.
A total of 1{,}000 queries are sampled from the H3 dataset~\cite{guo-etal-2023-hc3} and sent from a client located in Hong Kong. 

\subsubsection{TIQE Protocol Performance} We employ the Cross-Encoder model \emph{TinyLM-L6-v2}%
\footnote{\url{https://huggingface.co/cross-encoder/ms-marco-MiniLM-L6-v2}} and LLM-as-a-Judge using the \emph{DeepSeek} API%
\footnote{\url{https://api-docs.deepseek.com/}} (0.07 \$/1M input token, 1.10 \$/1M output token). For degradation detection, we sample 50 true and 50 false response examples from the ShareGPT-Chinese-English-90k dataset~\cite{ShareGPT-Chinese-English-90k}, which consists of real-world dialogues between users and ChatGPT-4o.

\subsubsection{Task Reward Distribution} To evaluate the proposed reward mechanism, we analyze reward distribution under varying conditions, with parameters summarized in Table~\ref{tab:reward_params}.

\begin{table}[htbp]
\centering
\caption{Evaluation Parameters Setting for Reward Distribution}
\label{tab:reward_params}
\begin{adjustbox}{width=\linewidth}
\begin{tabular}{lll}
\toprule
 \textbf{Parameter} & \textbf{Description} & \textbf{Value} \\
\midrule
 $\beta$ & Portion of reward reserved for validators & 0.3 \\
 $\theta$ & Scaling factor for task reward & 1.0 \\
 $b$ & Tasks per batch & 32 \\
$R_i$ & Cost per task & $[0.1, 0.5]$ \\
 $\alpha$ & Score representing result quality & \{0.2, 0.5, 0.8, 1.0\} \\
 $\bm{ST}$ & Validator stake vectors & \begin{tabular}[c]{@{}l@{}}Case 1: [100, 100, 100] \\ Case 2: [100, 300, 600]\end{tabular} \\
\bottomrule
\end{tabular}
\end{adjustbox}
\end{table}


\subsection{Results and Analysis}

\subsubsection{Geo-distributed Decentralized Inference Performance} 

Table~\ref{tab:performance} summarizes the inference performance. Model size significantly affects latency and throughput: models with few parameters (e.g., 1.5B) yield lower latency and higher throughput, while models with large parameters (e.g., 7B, 14B) incur substantial overhead. Geographical distance also impacts performance, workers in Hong Kong consistently achieve better results. For instance, the RTX 4060 Ti running the 1.5B model attains an average latency of \SI{10.82}{s}, TTFT of \SI{1.02}{s}, and throughput of \SI{120.40}{tok/s}; in contrast, the same model in Shenzhen shows worse performance. Cross-device inference enables execution of models with large parameters (e.g., 14B), but introduces additional overhead, with the highest latency observed at \SI{160.05}{s}. In all configurations, failure rates remain below 5\%, and acceptable latency for personal use. These results confirm that \textsc{PolyLink} supports reliable, responsive inference across heterogeneous, geo-distributed environments.

\begin{table*}[htbp]
    \centering
    \caption{Geo-distributed Decentralized Inference Performance in \textsc{PolyLink}.}
    \label{tab:performance}
    \resizebox{\textwidth}{!}{%
    \begin{threeparttable}
        \begin{tabular}{lccccccc}
            \toprule
                Device Location & Device Type & Model & Latency (s) & TTFT (s) & OTPS (tok/s) & RTPS (req/s) & FR\\
            \midrule
                \multirow{3}{*}{Hong Kong} & RTX4060 Ti & Deepseek-R1-1.5B & 10.82 & 1.02 & 120.40 & 0.092 & 0\%\\
                
                  & RTX3080 Ti & Deepseek-R1-7B & 12.13 & 1.20 & 87.30 & 0.082 & 0\%\\

                  & RTX4090 & Deepseek-R1-14B & 17.28 & 1.37 & 55.73 & 0.058 & 0\%\\
                  
                \multirow{2}{*}{Shenzhen} & Apple M3Pro & Deepseek-R1-1.5B & 23.30 & 5.65 & 49.46 & 0.041 & 2\% \\

                  & Apple M3Pro & Deepseek-R1-7B & 52.84 & 17.59 & 18.02 & 0.018 & 5\% \\
            \midrule
                Hong Kong  & RTX3080 Ti \& RTX4060 Ti & Deepseek-R1-14B & 160.05 & 8.84 & 7.12 & 0.006 & 2\% \\
            \bottomrule
            \end{tabular}  
        \begin{tablenotes}[flushleft]
            \footnotesize
            \item $\dagger$ Latency: The total time from sending a request to receiving the complete response. TTFT: The time from request initiation to the first token being received. OTPS: The number of output tokens generated per second by the model. RTPS: The number of complete requests processed per second. FR: The percentage of inference requests that did not complete successfully.
        
        \end{tablenotes}
    \end{threeparttable}
   }
\end{table*}

\begin{figure}[htbp]
    \centering
    \includegraphics[width=0.8\linewidth]{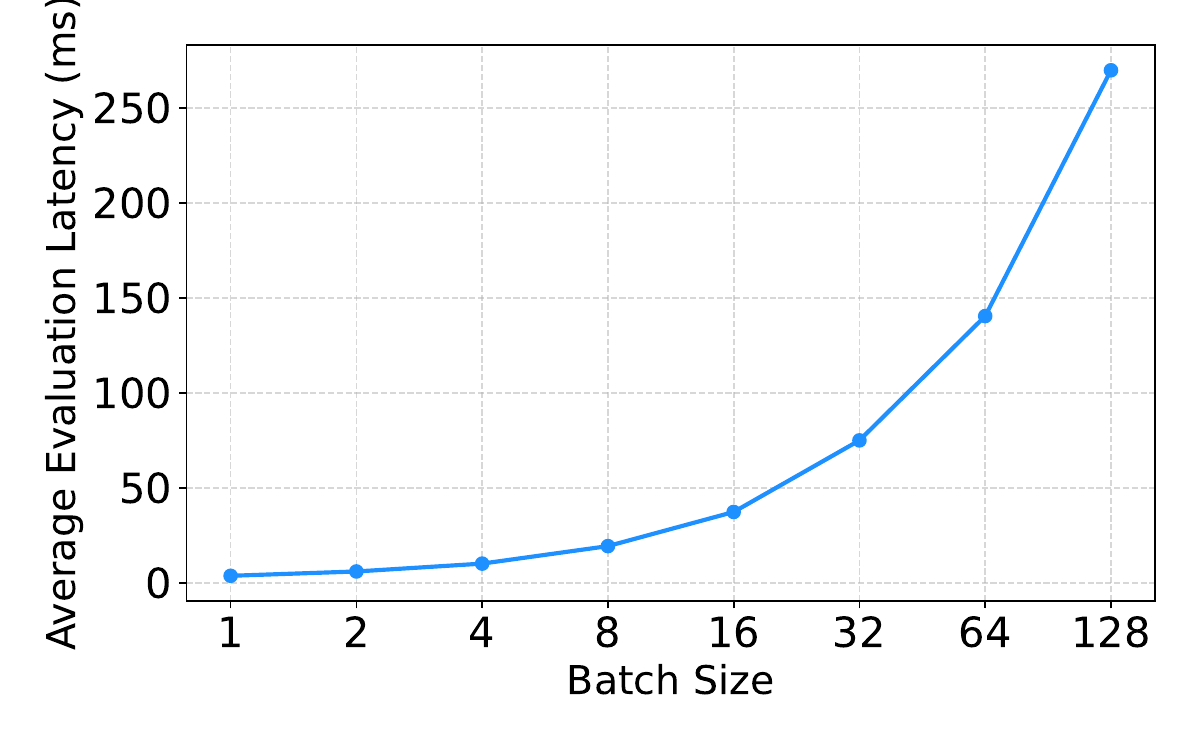}
    \caption{Average per-batch evaluation latency of the Cross-Encoder under different batch sizes.}
    \label{fig:cross-encoder}
\end{figure}

\begin{figure}[htbp]
  \centering
  \begin{subfigure}[t]{0.49\linewidth}
    \centering
    \includegraphics[width=\linewidth]{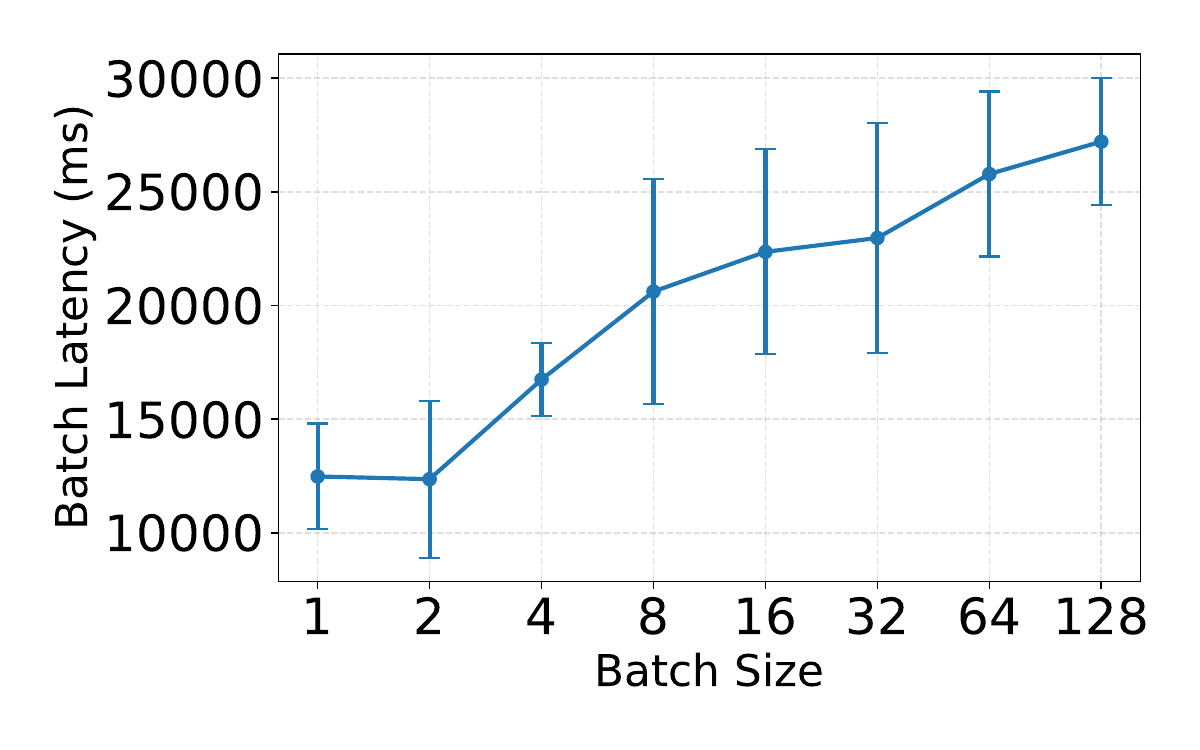}
    \caption{Latency}
    \label{fig:llm1}
  \end{subfigure}
  \hfill
  \begin{subfigure}[t]{0.49\linewidth}
    \centering
    \includegraphics[width=\linewidth]{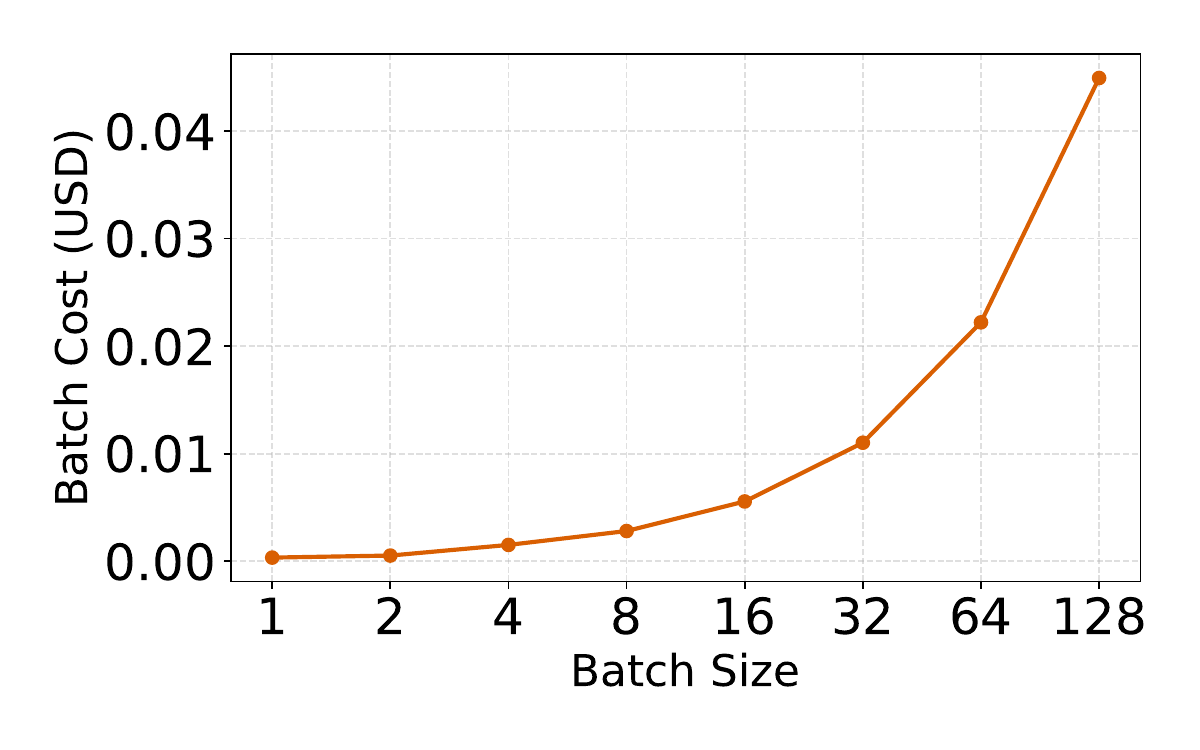}
    \caption{Cost}
    \label{fig:llm2}
  \end{subfigure}
  \caption{Per-Batch Latency and Cost of LLM-as-a-Judge under different batch sizes.}
  \label{fig:llm}
\end{figure}

\begin{figure}[htbp]
    \centering
    \includegraphics[width=1\linewidth]{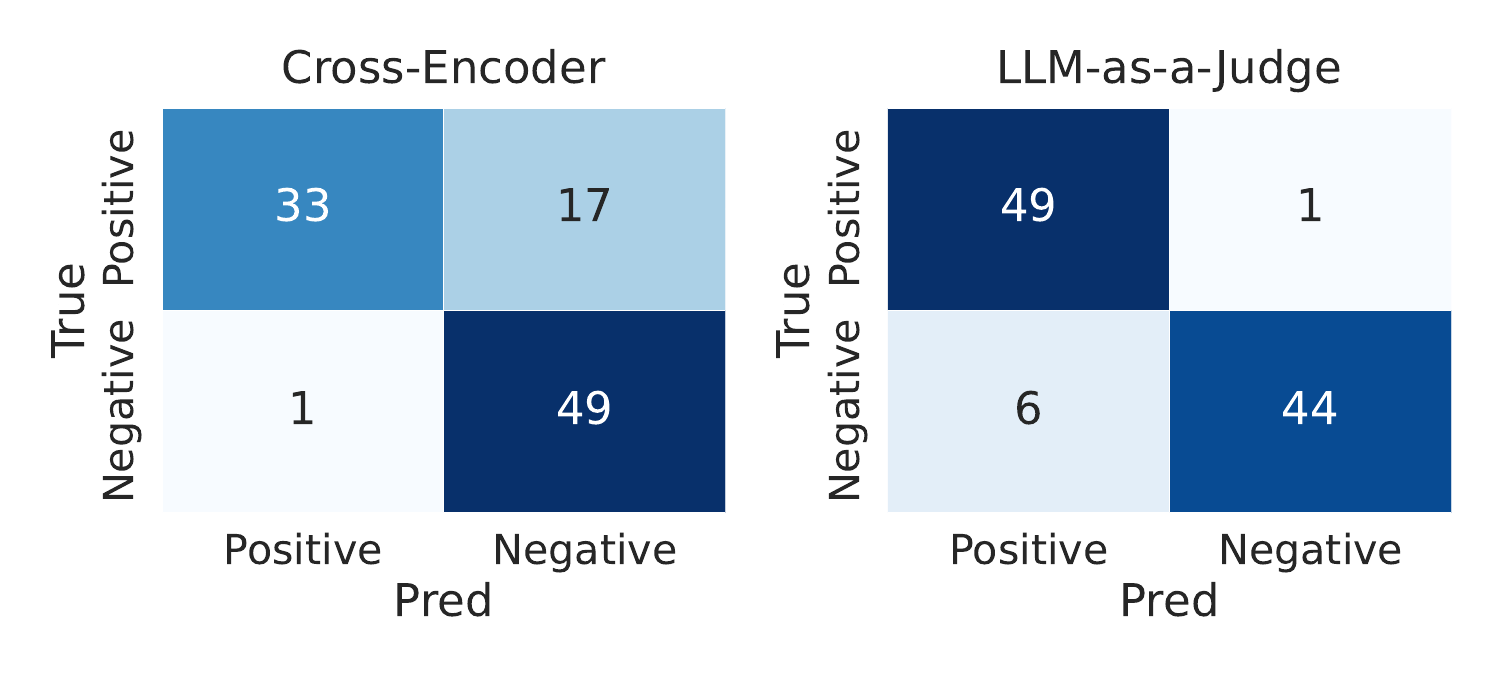}
    \caption{Confusion matrices of cross encoder and LLM-as-a-Judge for detecting Model Degradation Attack.}
    \label{fig:cm}
\end{figure}

\subsubsection{TIQE Protocol Performance}

 As shown in Fig.~\ref{fig:cross-encoder}, the average per-batch evaluation latency of Cross-Encoder increases with batch size, from around 10\,ms at batch size 1 to over 250\,ms at batch size 128. This indicates that the Cross-Encoder is lightweight and scales efficiently with moderate batch sizes. Since the model is lightweight, we omit its cost in our analysis.
 
 Fig.~\ref{fig:llm} shows the latency and cost of invoking LLM-as-a-Judge. Despite concurrent requests within each batch, latency increases with batch size due to queuing, and cost rises accordingly. This reveals a trade-off between batching efficiency and response time.
 
 Fig.~\ref{fig:cm} illustrates the confusion matrices of the Cross-Encoder and LLM-as-a-Judge on the degradation detection task. The Cross-Encoder achieves a TP rate of 66\% (33/50) and a FP rate of 2\% (1/50), indicating moderate detection capability with minimal false alarms. In contrast, the LLM-as-a-Judge significantly improves the TP rate to 98\% (49/50), with a slightly higher FP rate of 12\% (6/50).

 These results show that the Cross-Encoder offers low-latency, cost-effective evaluation with acceptable accuracy, while LLM-as-a-Judge provides higher detection accuracy with increased latency and cost.

\subsubsection{Task Reward Distribution}

 Fig. \ref{fig:rd} shows the reward distribution results. As the model quality score increases, the worker’s reward consistently rises, reflecting higher returns for higher-quality inference. In the unequal stake setting (Fig. \ref{fig:rd2}), the validator with a higher stake (Validator 3) receives the largest share among validators, demonstrating that rewards are allocated proportionally to stake. These results validate that the proposed mechanism fairly and effectively allocates rewards based on both contribution quality and economic stake.

\begin{figure}[htbp]
  \centering
  \begin{subfigure}[t]{0.49\linewidth}
    \centering
    \includegraphics[width=\linewidth]{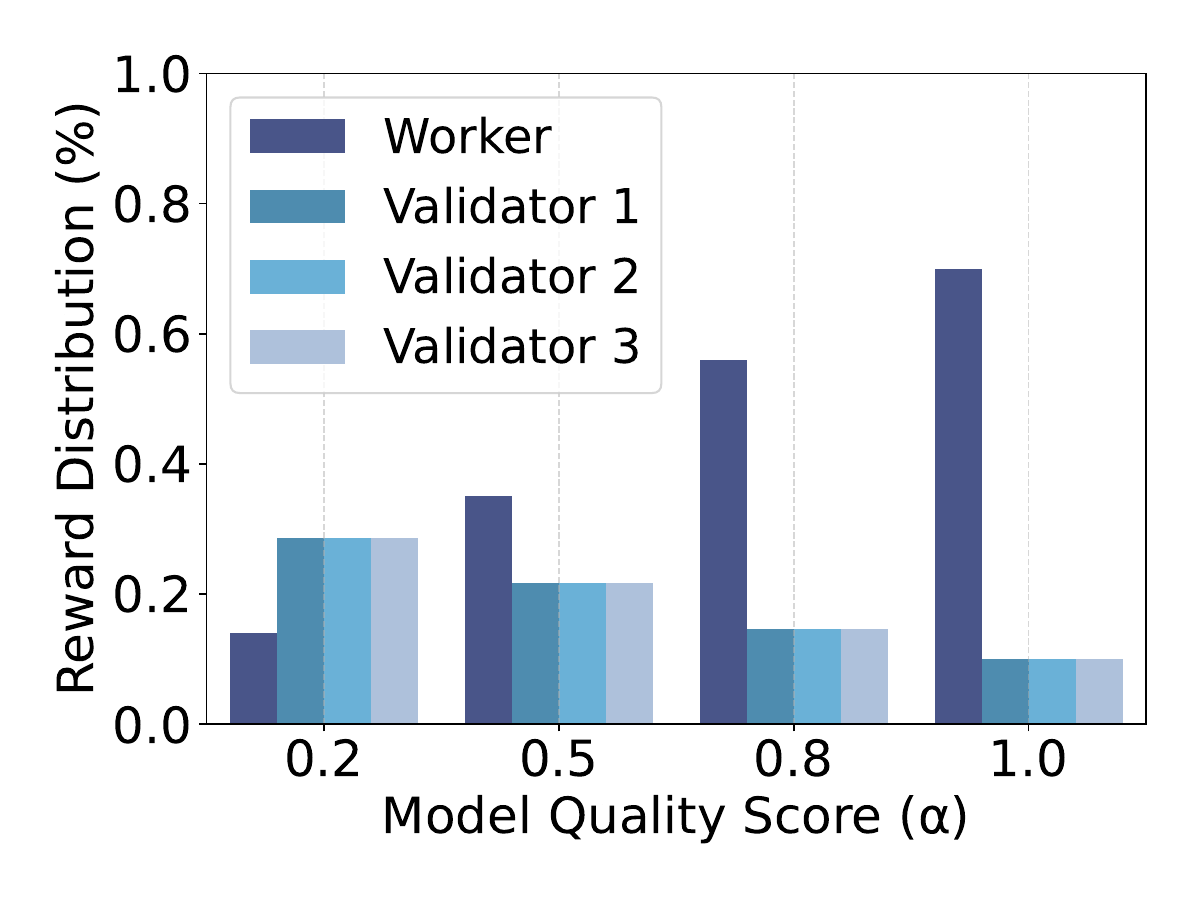}
    \caption{Equal Stake Setting.}
    \label{fig:rd1}
  \end{subfigure}
  \hfill
  \begin{subfigure}[t]{0.49\linewidth}
    \centering
    \includegraphics[width=\linewidth]{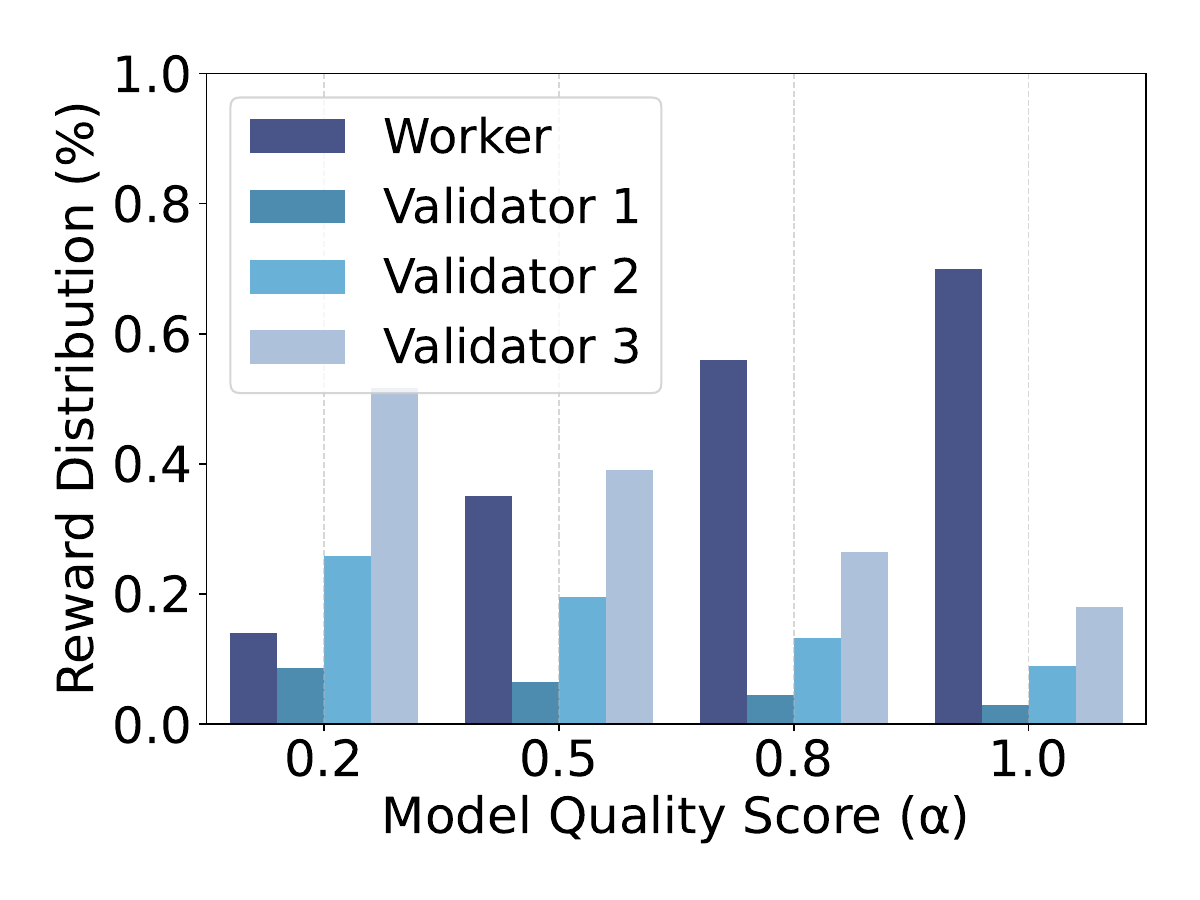}
    \caption{Unequal Stake Setting.}
    \label{fig:rd2}
  \end{subfigure}
  \caption{Reward distribution in a batch under different stake settings.}
  \label{fig:rd}
\end{figure}

\section{Conclusion and Future Works}
\noindent We proposed \textsc{PolyLink}, a blockchain-based decentralized AI platform that enables decentralized LLM inference across edge networks. To ensure inference integrity at low cost, we introduced the Trustless Inference Quality Evaluation (TIQE) protocol. In addition, our incentive model promotes fair and effective reward allocation among network participants. Real-world deployment and evaluation demonstrate that \textsc{PolyLink} is a practical and scalable solution for decentralized AI in Web3 and DePIN ecosystems. However, the platform still has some limitations. First, the security assumption that adversaries control less than one-third of the validators may not hold in practical environments. Second, the cross-device inference method suffers from network communication latency, so the number of devices is limited. In the future, we will explore the solution to the limitations and extend the network scale of PolyLink by incorporating cross-chain support \cite{cao2025map} and model training \cite{10.1145/3701716.3715484} \cite{10.1145/3696410.3714666}. Moreover, we plan to deploy smart city applications including digital twins and metaverse over PolyLink \cite{10294558} \cite{11020582}. 

\section*{Acknowledgment}

\noindent This work is supported by CM-PolyU Joint Research Project (No. R24114H7), HK RGC Theme-based Research Scheme (No. T43-513/23-N), and Research Institute for Artificial Intelligence of Things, The Hong Kong Polytechnic University.

\bibliographystyle{IEEEtran}
\bibliography{Ref}

\end{document}